\documentclass[twoside]{article}%
\usepackage{amssymb}
\usepackage{amsfonts}
\usepackage{amsmath}
\usepackage{graphicx}%
\providecommand{\U}[1]{\protect\rule{.1in}{.1in}}
\newtheorem{theorem}{Theorem}

\newtheorem{definition}[theorem]{Definition}

\newtheorem{lemma}[theorem]{Lemma}

\newtheorem{remark}[theorem]{Remark}

\newenvironment{proof}[1][Proof]{\noindent\textbf{#1.} }{\ \rule{0.5em}{0.5em}}
\begin{document}

\title{\textbf{Analytical Solutions to the Navier-Stokes Equations}}
\author{Y\textsc{uen} M\textsc{anwai\thanks{E-mail address: nevetsyuen@hotmail.com }}\\\textit{Department of Applied Mathematics, The Hong Kong Polytechnic
University,}\\\textit{Hung Hom, Kowloon, Hong Kong}}
\date{Revised 20-Feb-2009}
\maketitle

\begin{abstract}
With the previous results for the analytical blowup solutions of the
$N$-dimensional $(N\geq2)$ Euler-Poisson equations, we extend the similar
structure to construct an analytical family of solutions for the isothermal
Navier-Stokes equations and pressureless Navier-Stokes equations with
density-dependent viscosity.

\end{abstract}

\section{Introduction}

The Navier-Stokes equations can be formulated in the following form:%
\begin{equation}
\left\{
\begin{array}
[c]{rl}%
{\normalsize \rho}_{t}{\normalsize +\nabla\cdot(\rho u)} & {\normalsize =}%
{\normalsize 0,}\\
{\normalsize (\rho u)}_{t}{\normalsize +\nabla\cdot(\rho u\otimes
u)+\delta\nabla P} & {\normalsize =}vis(\rho,u).
\end{array}
\right.  \label{eq1}%
\end{equation}
As usual, $\rho=\rho(x,t)$ and $u(x,t)$ are the density, the velocity
respectively. $P=P(\rho)$ is the pressure. We use a $\gamma$-law on the
pressure, i.e.
\begin{equation}
P(\rho)=K\rho^{\gamma}, \label{eq2}%
\end{equation}
with $K>0$, which is a universal hypothesis. The constant $\gamma=c_{P}%
/c_{v}\geq1$, where $c_{p}$ and $c_{v}$ are the specific heats per unit mass
under constant pressure and constant volume respectively, is the ratio of the
specific heats. $\gamma$ is the adiabatic exponent in (\ref{eq2}). In
particular, the fluid is called isothermal if $\gamma=1$. It can be used for
constructing models with non-degenerate isothermal fluid. $\delta$ can be the
constant $0$ or $1$. When $\delta=0$, we call the system is pressureless; when
$\delta=1$, we call that it is with pressure. And $vis(\rho,u)$ is the
viscosity function. When $vis(\rho,u)=0$, the system (\ref{eq1}) becomes the
Euler equations. For the detailed study of the Euler and Navier-Stokes
equations, see \cite{CW} and \cite{L}. In the first part of this article, we
study the solutions of the $N$-dimensional $(N\geq1)$ isothermal equations in
radial symmetry:%
\begin{equation}
\left\{
\begin{array}
[c]{rl}%
\rho_{t}+u\rho_{r}+\rho u_{r}+{\normalsize \frac{N-1}{r}\rho u} &
{\normalsize =0,}\\
\rho\left(  u_{t}+uu_{r}\right)  +\nabla K\rho & {\normalsize =}vis(\rho,u).
\end{array}
\right.  \label{eq3}%
\end{equation}

\begin{definition}
[Blowup]We say a solution blows up if one of the following conditions is
satisfied:\newline(1)The solution becomes infinitely large at some point $x$
and some finite time $T$;\newline(2)The derivative of the solution becomes
infinitely large at some point $x$ and some finite time $T$.
\end{definition}

For the formation of singularity in the 3-dimensioanl case for the Euler
equations, please refer the paper of Sideris \cite{SI}. In this article, we
extend the results form the study of the (blowup) analytical solutions in the
$N$-dimensional $(N\geq2)$ Euler-Poisson equations, which describes the
evolution of the gaseous stars in astrophysics \cite{DXY}, \cite{GW},
\cite{M1}, \cite{Y} and \cite{Y1}, to the Navier-Stokes equations. For the
similar kinds of blowup results in the non-isothermal case of the Euler or
Navier-Stokes equations, please refer \cite{Li} and \cite{Y}.

Recently, Yuen's results in \cite{Y1}, there exists a family of the blowup
solution for the Euler-Poisson equations in the $2$-dimensional radial
symmetry case,%
\begin{equation}
\left\{
\begin{array}
[c]{rl}%
\rho_{t}+u\rho_{r}+\rho u_{r}+{\normalsize \frac{1}{r}\rho u} &
{\normalsize =0,}\\
\rho\left(  u_{t}+uu_{r}\right)  +K\rho_{r} & {\normalsize =-}\frac{2\pi\rho
}{r}\int_{0}^{r}\rho(t,s)sds.
\end{array}
\right.  \label{gamma=1}%
\end{equation}
The solutions are
\begin{equation}
\left\{
\begin{array}
[c]{c}%
\rho(t,r)=\frac{1}{a(t)^{2}}e^{y(r/a(t))}\text{, }{\normalsize u(t,r)=}%
\frac{\overset{\cdot}{a}(t)}{a(t)}{\normalsize r;}\\
\overset{\cdot\cdot}{a}(t){\normalsize =}-\frac{\lambda}{a(t)},\text{
}{\normalsize a(0)=a}_{0}>0{\normalsize ,}\text{ }\overset{\cdot}%
{a}(0){\normalsize =a}_{1};\\
\overset{\cdot\cdot}{y}(x){\normalsize +}\frac{1}{x}\overset{\cdot}%
{y}(x){\normalsize +\frac{2\pi}{K}e}^{y(x)}{\normalsize =\mu,}\text{
}y(0)=\alpha,\text{ }\overset{\cdot}{y}(0)=0,
\end{array}
\right.  \label{solution1}%
\end{equation}
where $K>0$, $\mu=2\lambda/K$ with a sufficiently small $\lambda$ and $\alpha$
are constants.\newline(1)When $\lambda>0$, the solutions blow up in a finite
time $T$;\newline(2)When $\lambda=0$, if $a_{1}<0$, the solutions blow up at
$t=-a_{0}/a_{1}$.

In this paper, we extend the above result to the isothermal Navier-Stokes
equations in radial symmetry with the usual viscous function
\[
vis(\rho,u)=v\Delta u,
\]
where $v$ is a positive constant:%
\begin{equation}
\left\{
\begin{array}
[c]{rl}%
\rho_{t}+u\rho_{r}+\rho u_{r}+{\normalsize \frac{N-1}{r}\rho u} &
{\normalsize =0,}\\
\rho\left(  u_{t}+uu_{r}\right)  +K\rho_{r} & {\normalsize =v(u}_{rr}%
+\frac{N-1}{r}u_{r}-\frac{N-1}{r^{2}}u),
\end{array}
\right.  \label{usualns1}%
\end{equation}

\begin{theorem}
\label{thm:1}For the $N$-dimensional isothermal Navier-Stokes equations in
radial symmetry (\ref{usualns1}), there exists a family of solutions, those
are:%
\begin{equation}
\left\{
\begin{array}
[c]{c}%
\rho(t,r)=\frac{1}{a(t)^{N}}e^{y(r/a(t))},u(t,r)=\frac{\overset{\cdot}{a}%
(t)}{a(t)}r,\\
\overset{\cdot\cdot}{a}(t)=\frac{-\lambda}{a(t)},a(0)=a_{0}>0,\overset{\cdot
}{a}(0)=a_{1},\\
y(x)=\frac{\lambda}{2K}x^{2}+\alpha,
\end{array}
\right.  \label{*4}%
\end{equation}
where $\alpha$ and $\lambda$ are arbitrary constants.\newline In particular,
for $\lambda>0$, the solutions blow up in finite time $T$.
\end{theorem}

In the last part, the corresponding solutions to the pressureless
Navier-Stokes equations with density-dependent viscosity is also studied.

\section{The Isothermal $(\gamma=1)$ Cases}

Before we present the proof of Theorem \ref{thm:1}, the Lemma 6 of \cite{Y1}
could be needed to further extended to the $N$-dimensional space.

\begin{lemma}
[The Extension of Lemma 6 of \cite{Y1}]\label{lem:generalsolutionformasseq}For
the equation of conservation of mass in radial symmetry:
\begin{equation}
\rho_{t}+u\rho_{r}+\rho u_{r}+\frac{N-1}{r}\rho u=0,
\label{massequationspherical}%
\end{equation}
there exist solutions,%
\begin{equation}
\rho(t,r)=\frac{f(r/a(t))}{a(t)^{N}},\text{ }{\normalsize u(t,r)=}%
\frac{\overset{\cdot}{a}(t)}{a(t)}{\normalsize r,}
\label{generalsolutionformassequation}%
\end{equation}
with the form $f\geq0\in C^{1}$ and $a(t)>0\in C^{1}.$
\end{lemma}

\begin{proof}
We just plug (\ref{generalsolutionformassequation}) into
(\ref{massequationspherical}). Then
\begin{align*}
&  \rho_{t}+u\rho_{r}+\rho u_{r}+\frac{N-1}{r}\rho u\\
&  =\frac{-N\overset{\cdot}{a}(t)f(r/a(t))}{a(t)^{N+1}}-\frac{\overset{\cdot
}{a}(t)r\overset{\cdot}{f}(r/a(t))}{a(t)^{N+2}}\\
&  +\frac{\overset{\cdot}{a}(t)r}{a(t)}\frac{\overset{\cdot}{f}(r/a(t))}%
{a(t)^{N+1}}+\frac{f(r/a(t))}{a(t)^{N}}\frac{\overset{\cdot}{a}(t)}%
{a(t)}+\frac{N-1}{r}\frac{f(r/a(t))}{a(t)^{N}}\frac{\overset{\cdot}{a}%
(t)}{a(t)}r\\
&  =0.
\end{align*}
The proof is completed.
\end{proof}

Besides, the Lemma 7 of \cite{Y1} is also useful. For the better understanding
of the lemma, the proof is given here.

\begin{lemma}
[lemma 7 of \cite{Y1}]$\label{lemma123}$For the Emden equation,%
\begin{equation}
\left\{
\begin{array}
[c]{c}%
\ddot{a}(t)=-\frac{\lambda}{a(t)},\\
a(0)=a_{0}>0,\ \dot{a}(0)=a_{1},
\end{array}
\right.  \label{LE1}%
\end{equation}
we have, if $\lambda>0$, there exists a finite time $T_{-}<+\infty$ such that
$a(T_{-})=0$.
\end{lemma}

\begin{proof}
By integrating (\ref{LE1}), we have%
\begin{equation}
0\leq\frac{1}{2}\overset{\cdot}{a}(t)^{2}=-\lambda\ln a(t)+\theta
\label{eq1114}%
\end{equation}
where $\theta=\lambda\ln a_{0}+\frac{1}{2}a_{1}^{2}.$\newline From
(\ref{eq1114}), we get,%
\[
a(t)\leq e^{\theta/\lambda}.
\]
If the statement is not true, we have%
\[
0<a(t)\leq e^{\theta/\lambda},\text{ for all }t\geq0.
\]
But since
\[
\ddot{a}(t)=-\frac{\lambda}{a(t)}\leq\frac{-\lambda}{e^{\theta/\lambda}},
\]
we integrate this twice to deduce%
\[
a(t)\leq\int_{0}^{t}\int_{0}^{\tau}\frac{-\lambda}{e^{\theta/\lambda}}%
dsd\tau+C_{1}t+C_{0}=\frac{-\lambda t^{2}}{2e^{\theta/\lambda}}+C_{1}t+C_{0}.
\]
By taking $t$ large enough, we get%
\[
a(t)<0.
\]
As a contradiction is met, the statement of the Lemma is true.
\end{proof}

By extending the structure of the solutions (\ref{solution1}) to the
2-dimensional isothermal Euler-Poisson equations (\ref{gamma=1}) in \cite{Y1},
it is a natural result to get the proof of the Theorem \ref{thm:1}.

\begin{proof}
[Proof of Theorem \ref{thm:1}]By using the Lemma
\ref{lem:generalsolutionformasseq}, we can get that (\ref{*4}) satisfy
(\ref{usualns1})$_{1}$. For the momentum equation, we have,%
\begin{align*}
&  \rho(u_{t}+u\cdot u_{r})+K\rho_{r}-v({\normalsize u}_{rr}+\frac{N-1}%
{r}u_{r}-\frac{N-1}{r^{2}}u)\\
&  =\rho\frac{\overset{\cdot\cdot}{a}(t)}{a(t)}r+\frac{K}{a(t)}\rho
\overset{\cdot}{y}(\frac{r}{a(t)})\\
&  =\frac{\rho}{a(t)}[-\frac{\lambda r}{a(t)}+K\overset{\cdot}{y}(\frac
{r}{a(t)})].
\end{align*}
By choosing
\[
y(x)=\frac{\lambda}{2K}x^{2}+\alpha,
\]
we have verified that (\ref{*4}) satisfies the above (\ref{usualns1})$_{2}$ .
If $\lambda>0$, by the Lemma \ref{lemma123}, there exists a finite time $T$
for such that $a(T_{-})=0$. Thus, there exist blowup solutions in finite time
$T$. The proof is completed.
\end{proof}

With the assistance of the blowup rate results of the Euler-Poisson equations
i.e. Theorem 3 in \cite{Y1}, it is trivial to have the following theorem:

\begin{theorem}
\label{thm:2}With $\lambda>0$, the blowup rate of the solutions (\ref{*4}) is,%
\[
\underset{t\rightarrow T_{\ast}}{\lim}\rho(t,0)(T_{\ast}-t)^{\alpha}\geq
O(1),
\]
where the blowup time $T_{\ast}$ and $\alpha<N$ are constants.
\end{theorem}

\begin{remark}
If we are interested in the mass of the solutions, the mass of the solutions
can be calculated by:.%
\[
M(t)=\int_{R^{N}}^{{}}\rho(t,s)ds=\alpha(N)\int_{0}^{+\infty}\rho
(t,s)s^{N-1}ds,
\]
where $\alpha(N)$ denotes some constant related to the unit ball in $R^{N}$:
$\alpha(1)=1$; $\alpha(2)=2\pi$; for $N\geq3,$%
\[
\alpha(N)=N(N-2)V(N)=N(N-2)\frac{\pi^{N/2}}{\Gamma(N/2+1)},
\]
where $V(N)$ is the volume of the unit ball in $R^{N}$ and $\Gamma$ is the
Gamma function. We observe that the mass of the initial time $0$:\newline(1)
for $\lambda\geq0$%
\[
M(0)=\frac{\alpha(N)}{a_{0}^{N}}\int_{0}^{+\infty}e^{\frac{\lambda}{2K}%
s^{2}+\alpha}s^{N-1}ds.
\]
The mass is infinitive. The very large density comes from the ends of outside
of the origin $O$.\newline(2) for $\lambda<0,$%
\[
M(0)=\frac{\alpha(N)}{a_{0}^{N}}\int_{0}^{+\infty}e^{\frac{\lambda}{2K}%
s^{2}+\alpha}s^{N-1}ds=\frac{\alpha(N)e^{\alpha}}{a_{0}^{N}}\int_{0}^{+\infty
}e^{\frac{\lambda}{2K}s^{2}}s^{N-1}ds.
\]
The mass of the solution can be arbitrarily small but without compact support
if $\alpha$ is taken to be a very small negative number.
\end{remark}

\begin{remark}
Our results can be easily extended to the isothermal Euler/ Navier-Stokes
equations with frictional damping term with the assistance of Lemma 7 in
\cite{Y}:%
\[
\left\{
\begin{array}
[c]{c}%
\rho_{t}+u\rho_{r}+\rho u_{r}+\frac{N-1}{r}\rho u=0,\\
\rho(u_{t}+u\cdot u_{r})+K\rho_{r}+\beta\rho u=v{\normalsize (u}_{rr}%
+\frac{N-1}{r}u_{r}-\frac{N-1}{r^{2}}u),
\end{array}
\right.
\]
where $\beta\geq0$ and $v\geq0$.\newline The solutions are:%
\[
\left\{
\begin{array}
[c]{c}%
\rho(t,r)=\frac{e^{y(r/a(t))}}{a(t)^{N}},u(t,r)=\frac{\overset{\cdot}{a}%
(t)}{a(t)}r,\\
\overset{\cdot\cdot}{a}(t)+\beta\dot{a}(t)=\frac{-\lambda}{a(t)}%
,a(0)=a_{0}>0,\overset{\cdot}{a}(0)=a_{1},\\
y(x)=\frac{\lambda}{2K}x^{2}+\alpha.
\end{array}
\right.
\]

\end{remark}

\begin{remark}
Our results can be easily extended to the isothermal Euler/ Navier-Stokes
equations with frictional damping term with the assistance of Lemma 7 in
\cite{Y}:%
\[
\left\{
\begin{array}
[c]{c}%
\rho_{t}+u\rho_{r}+\rho u_{r}+\frac{N-1}{r}\rho u=0,\\
\rho(u_{t}+u\cdot u_{r})+K\rho_{r}+\beta\rho u=v{\normalsize (u}_{rr}%
+\frac{N-1}{r}u_{r}-\frac{N-1}{r^{2}}u)
\end{array}
\right.
\]
where $\beta\geq0$ and $v\geq0$.\newline The solutions are:%
\[
\left\{
\begin{array}
[c]{c}%
\rho(t,r)=\frac{e^{y(r/a(t))}}{a(t)^{N}},u(t,r)=\frac{\overset{\cdot}{a}%
(t)}{a(t)}r,\\
\overset{\cdot\cdot}{a}(t)+\beta\dot{a}(t)=\frac{-\lambda}{a(t)}%
,a(0)=a_{0}>0,\overset{\cdot}{a}(0)=a_{1},\\
y(x)=\frac{\lambda}{2K}x^{2}+\alpha.
\end{array}
\right.
\]

\end{remark}

\begin{remark}
The solutions (\ref{solution1}) to the Euler-Poisson equations only work for
the $2$-dimensional case. But the solutions (\ref{*4}) to the Navier-Stokes
equations work for the $N$-dimensional $(N\geq1)$ case.
\end{remark}

\begin{remark}
We may extend the solutions to the $2$-dimensional Euler/Navier-Stokes
equations with a solid core \cite{Lin}:%
\[
\left\{
\begin{array}
[c]{rl}%
\rho_{t}+u\rho_{r}+\rho u_{r}+{\normalsize \frac{1}{r}\rho u} &
{\normalsize =0,}\\
\rho\left(  u_{t}+uu_{r}\right)  +K\rho_{r}+\beta\rho u & {\normalsize =}%
\frac{M_{0}}{r}+{\normalsize v(u}_{rr}+\frac{1}{r}u_{r}-\frac{1}{r^{2}}u),
\end{array}
\right.
\]
where $M_{0}>0$, there is a unit stationary solid core locating $[0,r_{0}]$,
where $r_{0}$ is a positive constant, surrounded by the distribution
density.\newline The corresponding solutions are:%
\[
\left\{
\begin{array}
[c]{c}%
\rho(t,r)=\frac{e^{y(r/a(t))}}{a(t)^{2}},\text{ }u(t,r)=\frac{\overset{\cdot
}{a}(t)}{a(t)}r,\text{ for }r>r_{0},\\
\overset{\cdot\cdot}{a}(t)+\beta\dot{a}(t)=\frac{-\lambda}{a(t)},\text{
}a(0)=a_{0}>0,\overset{\cdot}{a}(0)=a_{1},\\
y(x)=\frac{\lambda}{2K}x^{2}+M_{0}\ln x+\alpha,
\end{array}
\right.
\]
where $\alpha>\frac{-\lambda}{2K}$ is a constant.
\end{remark}

\section{Pressureless Navier-Stokes Equations with Density-dependent
Viscosity}

Now we consider the pressureless Navier-Stokes equations with
density-dependent viscosity:%
\[
vis(\rho,u)\doteq\bigtriangledown(\mu(\rho)\bigtriangledown\cdot u),
\]
in radial symmetry:%
\begin{equation}
\left(
\begin{array}
[c]{rl}%
\rho_{t}+u\rho_{r}+\rho u_{r}+{\normalsize \frac{N-1}{r}\rho u} &
{\normalsize =0,}\\
\rho\left(  u_{t}+uu_{r}\right)   & {\normalsize =(\mu(}\rho))_{r}(\frac
{N-1}{r}u+u_{r})+\mu(\rho)(u_{rr}+\frac{N-1}{r}u_{r}-\frac{N-1}{r^{2}}u),
\end{array}
\right.  \label{NSDens}%
\end{equation}
where $\mu(\rho)$ is a density-dependent viscosity function, which is usually
written as $\mu(\rho)\doteq\kappa\rho^{\theta}$ with the constants $\kappa,$
$\theta>0$. For the study of this kind of the above system, the readers may
refer \cite{M2}\cite{Ni}\cite{YZ}.

We can obtain the similar estimate about Lemma \ref{lemma123} to the following
ODE,%
\begin{equation}
\left\{
\begin{array}
[c]{c}%
\ddot{a}(t)=\frac{\lambda\overset{\cdot}{a}(t)}{a(t)^{2}},\\
a(0)=a_{0}>0,\ \dot{a}(0)=a_{1}\leq\frac{\lambda}{a_{0}}.
\end{array}
\right.  \label{Lane-Emden}%
\end{equation}

\begin{lemma}
$\label{lemma1}$For the ODE (\ref{Lane-Emden}), with $\lambda>0$, there exists
a finite time $T_{-}<+\infty$ such that $a(T_{-})=0$.
\end{lemma}

\begin{proof}
(1) If $a(t)>0$ and $\dot{a}(0)=a_{1}\leq\frac{\lambda}{a_{0}}$ for all time
$t$, by integrating (\ref{Lane-Emden}), we have%
\begin{equation}
\overset{\cdot}{a}(t)=-\frac{\lambda}{a(t)}-\frac{\lambda}{a_{0}}+a_{1}%
\leq-\frac{\lambda}{a(t)}. \label{123}%
\end{equation}
Take the integration for (\ref{123}):
\begin{align*}
\int_{0}^{t}a(s)\overset{\cdot}{a}(s)ds  &  \leq-\int_{0}^{t}\lambda ds,\\
\frac{1}{2}[a(t)]^{2}  &  \leq-\lambda t+\frac{1}{2}a_{0}^{2}.
\end{align*}
When $t$ is very large, we have%
\[
\frac{1}{2}[a(t)]^{2}\leq-1.
\]
A contradiction is met. The proof is completed.
\end{proof}

Here we present another lemma before proceeding to the next theorem.

\begin{lemma}
$\label{lemma1 copy(1)}$For the ODE%
\begin{equation}
\left\{
\begin{array}
[c]{c}%
\overset{\cdot}{y}(x)y(x)^{n}-\xi x=0,\\
y(0)=\alpha>0,n\neq-1,
\end{array}
\right.  \label{seperateODE}%
\end{equation}
where $\xi$ and $n$ are constants,\newline we have the solution
\[
y(x)=\sqrt[n+1]{\frac{1}{2}(n+1)\xi x^{2}+\alpha^{n+1}}.
\]

\end{lemma}

\begin{proof}
The above ODE (\ref{seperateODE}) may be solved by the separation method:%
\[
\overset{\cdot}{y}(x)y(x)^{n}-\xi x=0,
\]%
\[
\overset{\cdot}{y}(x)y(x)^{n}=\xi x.
\]
By taking the integration with respect to $x:$%
\[
\int_{0}^{x}\overset{\cdot}{y}(x)y(x)^{n}dx=\int_{0}^{x}\xi xdx,
\]
we have,%
\begin{equation}
\int_{0}^{x}y(x)^{n}d[y(x)]=\frac{1}{2}\xi x^{2}+C_{1}, \label{eq111}%
\end{equation}
where $C_{1}$ is a constant. \newline By integration by part, then the
identity becomes%
\[
y(x)^{n+1}-n\int_{0}^{x}y(x)^{n-1}\overset{\cdot}{y}(x)y(x)dx=\frac{1}{2}\xi
x^{2}+C_{1},
\]%
\[
y(x)^{n+1}-n\int_{0}^{x}\overset{\cdot}{y}(x)y(x)^{n}dx=\frac{1}{2}\xi
x^{2}+C_{1}.
\]
From the equation (\ref{eq111}), we can have the simple expression for $y(x)$:%
\[
y(x)^{n+1}-n(\frac{1}{2}\xi x^{2}+C_{1})=\frac{1}{2}\xi x^{2}+C_{1},
\]%
\[
y(x)^{n+1}=\frac{1}{2}(n+1)\xi x^{2}+C_{2},
\]
where $C_{2}=(n+1)C_{1}.$

By plugging into the initial condition for $y(0)$, we have%
\[
y(0)^{n+1}=\alpha^{n+1}=C_{2}.
\]
Thus, the solution is:%
\[
y(x)=\sqrt[n+1]{\frac{1}{2}(n+1)\xi x^{2}+\alpha^{n+1}}.
\]
The proof is completed.
\end{proof}

The family of the solution to the pressureless Navier-Stokes equations with
density-dependent viscosity:%
\begin{equation}
\left(
\begin{array}
[c]{rl}%
\rho_{t}+u\rho_{r}+\rho u_{r}+{\normalsize \frac{N-1}{r}\rho u} &
{\normalsize =0,}\\
\rho\left(  u_{t}+uu_{r}\right)   & {\normalsize =(\kappa\rho}^{\theta}%
)_{r}(\frac{N-1}{r}u+u_{r})+{\normalsize \kappa\rho}^{\theta}(u_{rr}%
+\frac{N-1}{r}u_{r}-\frac{N-1}{r^{2}}u),
\end{array}
\right.  \label{NNS}%
\end{equation}
is presented as the followings:

\begin{theorem}
\label{thm:3}For the pressureless Navier-Stokes equations with
density-dependent viscosity (\ref{NNS}) in radial symmetry, there exists a
family of solutions,\newline for $\theta=1$:%
\[
\left\{
\begin{array}
[c]{c}%
\rho(t,r)=\frac{e^{y(r/a(t))}}{a(t)^{N}},u(t,r)=\frac{\overset{\cdot}{a}%
(t)}{a(t)}r,\\
\overset{\cdot\cdot}{a}(t)=\frac{\lambda\overset{\cdot}{a}(t)}{a(t)^{2}%
},a(0)=a_{0}>0,\overset{\cdot}{a}(0)=a_{1},\\
y(x)=\frac{\lambda}{2N\kappa}x^{2}+\alpha,
\end{array}
\right.
\]
where $\alpha$ and $\lambda$ are arbitrary constants. \newline In particular,
for $\lambda>0$ and $a_{1}$ $\leq\frac{\lambda}{a_{0}}$, the solutions blow up
in finite time;\newline for $\theta\neq1$:%
\begin{equation}
\left\{
\begin{array}
[c]{c}%
\rho(t,r)=\left\{
\begin{array}
[c]{cc}%
\frac{^{y(r/a(t))}}{a(t)^{N}}, & \text{ for }y(\frac{r}{a(t)})\geq0;\\
0, & \text{for }y(\frac{r}{a(t)})<0
\end{array}
\right.  ,\text{ }u(t,r)=\frac{\overset{\cdot}{a}(t)}{a(t)}r,\\
\overset{\cdot\cdot}{a}(t)=\frac{-\lambda\overset{\cdot}{a}(t)}%
{a(t)^{^{N\theta-N+2}}},\text{ }a(0)=a_{0}>0,\text{ }\overset{\cdot}%
{a}(0)=a_{1},\\
y(x)=\sqrt[\theta-1]{\frac{1}{2}(\theta-1)\frac{-\lambda}{N\kappa\theta}%
x^{2}+\alpha^{\theta-1}},
\end{array}
\right.  \label{theather1}%
\end{equation}
where $\alpha>0$.
\end{theorem}

\begin{proof}
[Proof of Theorem \ref{thm:3}]To (\ref{NNS})$_{1}$, we may use Lemma
\ref{lem:generalsolutionformasseq} to check it.\newline For $\theta=1$,
(\ref{NNS})$_{2,}$ becomes:%
\begin{align}
&  \rho(u_{t}+u\cdot u_{r})-{\normalsize (}\kappa\rho)_{r}(\frac{N-1}%
{r}u+u_{r})-\kappa\rho_{r}(u_{rr}+\frac{N-1}{r}u_{r}-\frac{N-1}{r^{2}%
}u)\label{*1}\\
&  =\rho\frac{\overset{\cdot\cdot}{a}(t)}{a(t)}r-N(\frac{\kappa e^{y(r/a(t))}%
}{a(t)^{N}})_{r}\frac{\overset{\cdot}{a}(t)}{a(t)}\nonumber\\
&  =\rho\left(  \frac{\lambda\overset{\cdot}{a}(t)r}{a(t)^{3}}\right)
-\frac{N\kappa e^{y(r/a(t))}\overset{\cdot}{y}(\frac{r}{a(t)})}{a(t)^{N+1}%
}\frac{\overset{\cdot}{a}(t)}{a(t)}\nonumber\\
&  =\frac{\rho\overset{\cdot}{a}(t)}{a(t)^{2}}\left(  \frac{\lambda r}%
{a(t)}-N\kappa\overset{\cdot}{y}(\frac{r}{a(t)})\right)  ,\nonumber
\end{align}
where we use
\[
\overset{\cdot\cdot}{a}(t)=\frac{\lambda\overset{\cdot}{a}(t)}{a(t)^{2}}.
\]
By choosing
\[
y(\frac{r}{a(t)})\doteq y(x)=\frac{\lambda}{2N\kappa}x^{2}+\alpha,
\]
(\ref{*1}) is equal to zero.\newline For the case of $\theta\neq1$,
(\ref{NNS})$_{2}$ can be calculated:
\begin{align}
&  \rho(u_{t}+u\cdot u_{r})-{\normalsize (}\kappa\rho^{\theta})_{r}\left(
\frac{N-1}{r}u+u_{r}\right)  -\kappa\rho^{\theta}(u_{rr}+\frac{N-1}{r}%
u_{r}-\frac{N-1}{r^{2}}u)\\
&  =\rho\left(  -\frac{\lambda\overset{\cdot}{a}(t)r}{a(t)^{N\theta-N+2}%
a(t)}\right)  -\frac{N\kappa\theta y(\frac{r}{a(t)})^{\theta-1}\overset{\cdot
}{y}(\frac{r}{a(t)})}{a(t)^{N(\theta-1)}a(t)^{N+1}}\frac{\overset{\cdot}%
{a}(t)}{a(t)}\nonumber\\
&  =\rho\left(  -\frac{\lambda\overset{\cdot}{a}(t)r}{a(t)^{N\theta-N+2}%
a(t)}\right)  -\frac{N\kappa\theta y(\frac{r}{a(t)})y(\frac{r}{a(t)}%
)^{\theta-2}\overset{\cdot}{y}(\frac{r}{a(t)})\overset{\cdot}{a}(t)}%
{a(t)^{N}a(t)^{N\theta-N+2}}\nonumber\\
&  =\rho\left(  -\frac{\lambda\overset{\cdot}{a}(t)r}{a(t)^{N\theta-N+2}%
a(t)}\right)  -\frac{N\kappa\theta\rho y(\frac{r}{a(t)})^{\theta-2}%
\overset{\cdot}{y}(\frac{r}{a(t)})\overset{\cdot}{a}(t)}{a(t)^{N\theta-N+2}%
}\nonumber\\
&  =\frac{-\rho\overset{\cdot}{a}(t)}{a(t)^{^{N\theta-N+2}}}\left(
-\frac{\lambda r}{a(t)}+N\kappa\theta y(\frac{r}{a(t)})^{\theta-2}%
\overset{\cdot}{y}(\frac{r}{a(t)})\right)  .
\end{align}
Define $x\doteq\frac{r}{a(t)}$, $n\doteq\theta-2,$ it follows:
\begin{align}
&  =\frac{-\rho\overset{\cdot}{a}(t)}{a(t)^{^{N\theta-N+2}}}\left(  \lambda
x+N\kappa\theta y(x)^{n}\overset{\cdot}{y}(x)\right)  \label{eq103}\\
&  =\frac{-\lambda\rho\overset{\cdot}{a}(t)}{a(t)^{^{N\theta-N+2}}}\left(
x+\frac{N\kappa\theta}{\lambda}y(x)^{n}\overset{\cdot}{y}(x)\right)  ,
\end{align}
and $\xi\doteq\frac{\lambda}{N\kappa\theta}$ in Lemma \ref{lemma1 copy(1)},
and choose%
\[
y(\frac{r}{a(t)})\doteq y(x)=\sqrt[\theta-1]{\frac{1}{2}(\theta-1)\frac
{-\lambda}{N\kappa\theta}x^{2}+\alpha^{\theta-1}}.
\]
And this is easy to check that%
\[
\overset{\cdot}{y}(0)=0.
\]
The equation (\ref{eq103}) is equal to zero. The proof is completed.
\end{proof}

\begin{remark}
By controlling the initial conditions in some solutions (\ref{theather1}), we
may get the blowup solutions. And the modified solutions can be extended to
the system in radial symmetry with frictional damping:%
\[
\left\{
\begin{array}
[c]{rl}%
\rho_{t}+u\rho_{r}+\rho u_{r}+{\normalsize \frac{N-1}{r}\rho u} &
{\normalsize =0,}\\
\rho\left(  u_{t}+uu_{r}\right)  +\beta\rho u & {\normalsize =(\mu(}\rho
))_{r}\left(  \frac{N-1}{r}u+u_{r}\right)  +\mu(\rho)(u_{rr}+\frac{N-1}%
{r}u_{r}-\frac{N-1}{r^{2}}u),
\end{array}
\right.
\]
where $\beta>0$,\newline with the assistance of the ODE:%
\[
\left\{
\begin{array}
[c]{c}%
\overset{\cdot\cdot}{a}(t)+\beta\overset{\cdot}{a}(t)=\frac{-\lambda
\overset{\cdot}{a}(t)}{a(t)^{S}},\\
a(0)=a_{0}>0,\overset{\cdot}{a}(0)=a_{1},
\end{array}
\right.
\]
where $S$ is a constant.
\end{remark}

\end{document}